\documentclass[journal]{IEEEtran}

\usepackage{graphicx}
\usepackage{amsfonts}
\usepackage{amsmath}
\usepackage{amssymb}
\usepackage{mathrsfs} 
\usepackage{color}
\usepackage{bm}

\newtheorem{thm}{Theorem}

\newtheorem{lem}{Lemma}
\newtheorem{proof}{proof}

\newtheorem{defn}{Definition}
\newtheorem{rem}{Remark}

\ifCLASSINFOpdf \else \fi 
\begin{document}
\title{Analog Error Correcting Codes with Constant Redundancy}

\author{ Wentu~Song, and
         ~Kui~Cai,~\IEEEmembership{Senior Member,~IEEE}
\thanks{Wentu~Song and Kui~Cai are with the Science, Mathematics and Technology
        Cluster, Singapore University of Technology and Design,
        Singapore 487372 (e-mail:
        \{wentu\_song, cai\_kui\}@sutd.edu.sg).
        \emph{Corresponding author: Kui Cai.}}
}

\maketitle

\begin{abstract}
We consider analog error-correcting codes (analog ECCs) that are
designed to correct/detect outlying errors arising in analog
implementations of vector-matrix multiplication. The
error-correction/detection capability of an analog ECC can be
characterized by its height profile, which is expected to be as
small as possible. In this paper, we consider analog ECCs whose
parity check matrix has columns of unit Euclidean norm. We first
present an upper bound on the height profile of such codes as well
as a simple decoder for correcting a single error. We then
construct a family of single error-correcting analog ECCs with
redundancy three for any code length, which have smaller height
profile compared to the known $[n,n-2]$ MDS constructions.
\end{abstract}


\IEEEpeerreviewmaketitle

\section{Introduction}

Analog computing has recently emerged as a promising paradigm for
accelerating linear algebraic operations, particularly
vector-matrix and matrix-matrix multiplications, by exploiting the
physical laws of emerging hardware such as resistive crossbar
arrays \cite{Boser91}$-$\cite{Dupraz23}. However, the intrinsic
analog nature of these devices also introduces non-negligible
computational inaccuracies, stemming from device variability,
noise, defects, and limited precision, which fundamentally
challenges the reliability of analog computation. In
\cite{Roth20}, Roth studied a class of codes, called \emph{analog
error correcting codes}, to handle outlying errors in
vector-matrix multiplication.

Consider the computing task of the multiplication of an
$\ell$-dimensional row vector $\bm u$ and an $\ell\times k$ matrix
$A'$, both over the real field $\mathbb R$, and the desired result
is the vector $\bm c'=\bm u A'$. To correct/detect possible
error(s), an extended vector $\bm c=\bm u[A',A'']=[\bm u A',\bm u
A'']$ is computed, where $A''$ is an $\ell\times r$ matrix which
should be carefully designed for error-correction/detection.
Given any $[n,k]$ linear code $\mathcal C$ over $\mathbb R$, where
$n=k+r$, the matrix $A''$ can be obtained from $\mathcal C$ as
follows. Let $G=[I_k,G']$ be a systematic generator matrix of
$\mathcal C$, where $I_k$ is the $k\times k$ identity matrix. Let
$A''=A'G'$. Then $\bm c=\bm u[A',A'']=\bm u[A',A'G']=\bm uA'G$ is
a codeword of $\mathcal C$. As such, it suffices to consider
linear codes over $\mathbb R$.

An $[n,k]$ linear code $\mathcal C$ over the real field $\mathbb
R$ is a $k$-dimensional subspace of $\mathbb R^n$. For any
codeword $\bm c\in\mathcal C$ to be computed, as errors occur, the
received vector is $\bm y=\bm c+\bm\varepsilon+\bm e$, where
$\bm\varepsilon=(\varepsilon_0,\varepsilon_1,\cdots,
\varepsilon_{n-1})\in\mathbb R^n$ represents the unavoidable and
tolerable small disturbances whose entries are within the interval
$[-\delta,\delta]$ for some prescribed positive $\delta$, and $\bm
e=(e_0,e_1,\cdots, e_{n-1})\in\mathbb R^n$ represents the outlying
errors. The error $e_j$ with $|e_j|>\Delta$ must be located, where
$\Delta>\delta$ is a prescribed positive real number
\cite{Roth20}.

The height profile of any linear code $\mathcal C$ was introduced
in \cite{Roth20} to characterize the error-correction capability
of $\mathcal C$. Specifically, for any
$m\in[n\rangle=\{0,1,\cdots,n-1\}$, the $m$-height of a codeword
in $\mathcal C$ is defined as the ratio between the largest and
the $(m+1)$th largest absolute values of its entries, and the
$m$-height of $\mathcal C$, denoted by $\mathsf{h}_m(\mathcal C)$,
is defined as the maximum of the $m$-height of all its codewords.
In \cite{Roth20}, it was shown that $\mathcal C$ can correct
$\tau$ errors and detect $\sigma$ additional errors with respect
to the threshold pair $(\delta, \Delta)$ if and only if
$2(\mathsf{h}_m(\mathcal C)+1)\leq\frac{\Delta}{\delta}$, where
$m=2\tau+\sigma$. In practice, to reduce the ``gray area'' between
$\delta$ and $\Delta$, the quantity $\Gamma_m(\mathcal
C)\triangleq 2(\mathsf{h}_m(\mathcal C)+1)$ is expected to be
small. In the design of $[n,k]$ linear analog error correcting
code $\mathcal C$, the first important requirement is that
$\Gamma_m(\mathcal C)$ should be as small as possible, for given
$n$, $k$ and $m$. Another important requirement is that the
decoding process should be simple to achieve high computational
throughput with low energy consumption.

It was shown in \cite{Roth20} that $\Gamma_m(\mathcal C)\geq 4$
for any $[n,k]$ linear code $\mathcal C$ and any $m\in[n\rangle$,
and for the case that $(m+1)|n$, $\Gamma_m(\mathcal C)=4$ if and
only if $\mathcal C$ is the $(n/(m+1))$-fold Cartesian power of
the $[m+1,1]$ repetition code. For single-error-detection $($i.e.,
$m=1)$, a family of $[n,k]$ linear codes with $\Gamma_1(\mathcal
C)\leq 2\cdot\lceil n/r\rceil$ was constructed in \cite{Roth20}
for any $1\leq k<n$, where $r=n-k$ is called the \emph{redundancy}
of the code for convenience. For single-error-correction $($i.e.,
$m=2)$, there are three constructions of $[n,k]$ linear codes
presented in \cite{Roth20}, which are: 1) codes with
$r=O(\sqrt{n})$ and $\Gamma_2(\mathcal C)\leq O(n/r)$; 2) codes
with $r=O(\log n)$ and $\Gamma_2(\mathcal C)=O(n/\sqrt{r})$; and
3) MDS codes with $r=2$ and $\Gamma_2(\mathcal C)=O(n^2)$. Single
error correcting MDS codes with $r=2$ and $\Gamma_2(\mathcal
C)=O(n^2)$ are also constructed in \cite{Zhengyi25} and an
effective decoding algorithm was also given in the same paper. A
summary of the known single error correcting $[n,k]$ linear codes
is given in Table 1.

\begin{table}[htbp]
\small
\begin{center}
\renewcommand\arraystretch{1.4}
\begin{tabular}{|p{2.7cm}|p{2.0cm}|p{2.6cm}|}
\hline      ~~~~~~~Reference               &~~~~~~~~$r$    &~~~~~~~$\Gamma_2(\mathcal C)$       \\
\hline      Proposition 6 of \cite{Roth20} &$r=O(\sqrt{n})$&$\Gamma_2(\mathcal C)\leq O(n/r)$   \\
\hline      Section VI of \cite{Roth20}    &$r=O(\log_2 n)$  &$\Gamma_2(\mathcal C)=O(n/\sqrt{r})$\\
\hline      Proposition 11 of \cite{Roth20}&$r=2$          &$\Gamma_2(\mathcal C)=O(n^2)$       \\
\hline      \cite{Zhengyi25}               &$r=2$          &$\Gamma_2(\mathcal C)=O(n^2)$       \\
\hline      This work                      &$r=3$          &$\Gamma_2(\mathcal C)=O(n\sqrt{n})$ \\
\hline
\end{tabular}
\end{center}
\vspace{0mm} Table 1. Summary of single error correcting $[n,k]$
linear codes $\mathcal C$, where $r=n-k$ is the code redundancy.
\end{table}

Construction of $[n,k]$ linear codes correcting multiple errors
$($i.e., $m>2)$ are studied in \cite{Wei24}. The redundancy of the
constructions are either $O(\sqrt{n})$ or $O(\log_2 n)$. The
problem of computing the $m$-height of a code was studied in
\cite{Jiang24}, where a linear programming method of the problem
was presented. In \cite{Jiang24}, a family of $[n=k!, k]$ linear
codes was constructed based on permutations, named permutation
analog codes. Although the time complexity for determining the
$m$-heights of such codes is relatively low, the code rate is only
$1/(k-1)!$. Codes based on the icosahedron and dodecahedron were
considered in \cite{Ziyuan26} and their $m$-heights were analyzed.
The notion of the height profile was further studied in
\cite{Roth26}.

In this paper, we study analog ECCs whose parity-check matrices
have columns of unit Euclidean norm. We first establish an upper
bound on $\Gamma_m(\mathcal C)$ of such codes $\mathcal C$ for any
$1\leq m\leq r$, where $r$ is the redundancy of $\mathcal C$.
$($By definition, $\Gamma_m(\mathcal C)=\infty$ when $m>r$.$)$ We
also give a simple decoder of such codes for correcting a single
error. We then construct a family of $[n,k=n-3]$ linear codes with
$\Gamma_2(\mathcal C)=O(n\sqrt{n})$. Compared with the known
single-error-correcting $[n,k=n-2]$ MDS linear codes, our
construction achieves a reduction in $\Gamma_2(\mathcal C)$ by a
factor of $\sqrt{n}$, at the cost of one additional redundancy. A
comparison between our construction and existing works is provided
in Table 1.

\section{Preliminaries}
For any integers $\ell\leq n$, let $[\ell:n]=\{\ell,
\ell+1,\cdots,n-1\}$ and $[n\rangle=[0:n]=\{0,1,\cdots,n-1\}$. If
$A$ is a set, then $|A|$ is the size of $A$ and $2^{A}$ is the
collection of all subsets of $A$. Let $\mathbb R$ be the set of
real numbers, $\mathbb R^+$ be the set of positive real numbers
and $\mathbb R_{\geq 0}$ be the set of non-negative real numbers.
We use $\bm 0$ to denote the all-zero vector of any length. If $H$
is an $r\times n$ matrix over $\mathbb R$, then its columns are
denoted by $\bm h_j, j\in[n\rangle$, i.e., $\bm h_j$ is the
$(j+1)$th column of $H$. Unless otherwise specified, any vector
$\bm x\in\mathbb R^n$ is written as a row vector $\bm x=(x_0,
x_1,\cdots,x_{n-1})$ and the transpose of $\bm x$ is denoted by
${\bm x}^\top$, where $x_j$ is the $(j+1)$th component of $\bm x$
for each $j\in[n\rangle$. The dot product of two vectors $\bm
x=(x_0, x_1,\cdots,x_{n-1})$ and $\bm z=(z_0,
z_1,\cdots,z_{n-1})$, denoted by $\langle\bm x,\bm z\rangle$ is
defined as $\langle\bm x,\bm z\rangle=\bm x\bm
z^\top=\sum_{i=0}^{n-1}x_iz_i~($if $\bm x, \bm z$ are column
vectors, then $\langle\bm x,\bm z\rangle=\bm x^\top\bm z)$, and
the Euclidean norm of $\bm x$ is $\|\bm x\|_2=\sqrt{\langle\bm
x,\bm x\rangle}$. For any real number $a\geq 0$, let
$$\mathsf{Supp}_a(\bm x) \triangleq\left\{j\in[n\rangle:
|x_j|>a\right\}.$$ Note that $\mathsf{Supp}(\bm
x)\triangleq\mathsf{Supp}_0(\bm x)$ is just the ordinary support
of $\bm x$. The Hamming weight of $\bm x$, denoted by
$\mathsf{w}(\bm x)$, is defined as $|\mathsf{Supp}(\bm x)|$. For
positive integers $m\leq n$, let
$$\mathcal B(n,m)\triangleq\left\{\bm e\in\mathbb
R^n: \mathsf{w}(\bm e)\leq m\right\},$$ that is, $\mathcal B(n,w)$
is the set of all vectors in $\mathbb R^n$ of Hamming weight at
most $w$. For any $\delta\in\mathbb R^+$, let
$$\mathcal Q(n,\delta)\triangleq\left\{\bm \varepsilon\in\mathbb
R^n: \mathsf{Supp}_\delta(\bm \varepsilon)=\emptyset\right\}.$$ In
other words, $\mathcal Q(n,\delta)$ is the set of all vectors
$\bm\varepsilon=(\varepsilon_0,\varepsilon_1,\cdots,
\varepsilon_{n-1})\in\mathbb R^n$ such that
$|\varepsilon_i|\leq\delta$ for all $i\in[n\rangle$.

The following concept about analog error correcting codes was
introduced in \cite{Roth20,Roth23}.

\begin{defn}\label{def-anl-dec}
Let $\mathcal C$ be an $[n,k]$ linear code over $\mathbb R$. We
say that $\mathcal C$ corrects $\tau$ errors and detects $\sigma$
additional errors with respect to the threshold pair $(\delta,
\Delta)$ if there exists a decoder $\mathcal D:\mathbb
R^n\rightarrow 2^{[n\rangle}\cup\{\text{`}\mathsf{e}\text{'}\}$
such that for every $\bm y=\bm c+\bm\varepsilon+\bm e$, where $\bm
c\in\mathcal C$, $\bm\varepsilon\in\mathcal Q(n,\delta)$ and $\bm
e\in\mathcal B(n,\tau+\sigma)$, the following two conditions are
satisfied:
\begin{enumerate}
 \item[(D1)] If $\bm e\in\mathcal B(n,\tau)$, then
 $\text{`}\mathsf{e}\text{'}\neq\mathcal D(\bm y)\subseteq \mathsf{Supp}(\bm e)$.
 \item[(D2)] If $\mathcal D(\bm y)\neq\text{`}\mathsf{e}\text{'}$, then
 $\mathsf{Supp}_\Delta(\bm e)\subseteq\mathcal D(\bm y)$.
\end{enumerate}
\end{defn}

It was shown in \cite{Roth20} that the error correction/detection
capability of a code can be characterized by its height profile,
which can be defined as follows. Let $\mathcal C$ be a linear code
over $\mathbb R$ of length $n$ and $\bm 0\neq\bm c\in\mathcal C$
such that the entries of $\bm c$ are sorted according to
descending absolute values as:
$|c_{\pi(0)}|\geq|c_{\pi(1)}|\geq\cdots|c_{\pi(n-1)}|.$ For
$m\in[n\rangle$, the $m$-height of $\bm c$ is defined as
$\mathsf{h}_m(\bm c)\triangleq\frac{|c_{\pi(0)}|}{|c_{\pi(m)}|},$
and let $\mathsf{h}_m(\bm c)=\infty$ when $m\geq n$. If $\bm c=\bm
0$, define $\mathsf{h}_m(\bm c)=0$ for every $m\geq 0$. Further,
the $m$-height of $\mathcal C$ is defined as
$$\mathsf{h}_m(\mathcal C)\triangleq\max_{\bm c\in\mathcal
C}\mathsf{h}_m(\bm c).$$

For every $\bm c\neq \bm 0$, by the definition, $\mathsf{h}_m(\bm
c)=\infty$ if and only if $m\geq \mathsf{w}(\bm c)$, where
$\mathsf{w}(\bm c)$ is the Hamming weight of $\bm c$. Therefore,
$\mathsf{h}_m(\mathcal C)=\infty$ if and only if
$m\geq\mathsf{d}(\mathcal C)$, where $\mathsf{d}(\mathcal C)$ is
the minimum Hamming distance of $\mathcal C$.

\begin{lem}\cite[Theorem 1]{Roth20}\label{Err-C-Ablt}
A linear code $\mathcal C$ can correct $\tau$ errors and detect
$\sigma$ additional errors with respect to the threshold pair
$(\delta, \Delta)$ if and only if
$2(\mathsf{h}_{2\tau+\sigma}(\mathcal
C)+1)\leq\frac{\Delta}{\delta}$, or equivalently,
$\Gamma_{2\tau+\sigma}(\mathcal C)\leq\frac{\Delta}{\delta}$,
where $\Gamma_m(\mathcal C)\triangleq 2(\mathsf{h}_m(\mathcal
C)+1)$ for each $m\in[n\rangle$.
\end{lem}

By Definition \ref{def-anl-dec} and Lemma \ref{Err-C-Ablt}, the
error correction/detection capability of a code $\mathcal C$ is
represented by the ratio $\Delta/\delta$. In general, smaller
value of the ratio $\Delta/\delta$ means smaller ``gray area''
between the values of the outlying errors and the tolerable
errors. On the other hand, by Lemma \ref{Err-C-Ablt}, the minimum
value of $\Delta/\delta$ is determined by $\Gamma_{m}(\mathcal
C)$. Hence, we are interested in designing $[n,k]$ linear codes
$\mathcal C$ over $\mathbb R$ with $\Gamma_{m}(\mathcal C)$ as
small as possible for given $n$, $k$ and $m\leq n-k$. We only need
to consider $m\leq n-k$ because $\mathsf{h}_m(\mathcal C)=\infty$
for $m\geq\mathsf{d}(\mathcal C)$ and $\mathsf{d}(\mathcal C)\leq
n-k+1$ by the Singleton bound.

\section{A class of analog error correcting codes}

In this section, we consider a particular class of $[n,k]$ codes,
namely those whose parity-check matrices have columns of unit
Euclidean norm. We will give an upper bound on the $m$-height of
such codes for $1\leq m\leq n-k$, and we also present a simple
decoder for correcting a single error.

To estimate the height profile of codes, we need the following
notations and lemmas.

Let $H$ be an $r\times n$ matrix over $\mathbb R$. For the sake of
convenience, for any $A\subseteq\mathbb R^n$, we denote
$H*A=\{H{\bm a}^\top: \bm a\in A\}$. Then we let $\mathcal
S_H=H*\mathcal Q(n,1)$ and $2\mathcal S_H=\mathcal S_H+\mathcal
S_H$. Clearly, we have $2\mathcal S_H=H*\mathcal Q(n,2)$. For
$\Delta\in\mathbb R^+$, let
\begin{align*}\mathcal B_\Delta(n,m)&\triangleq\{\bm e\in\mathcal B(n,m):
\mathsf{Supp}_\Delta(\bm e)\neq\emptyset\}
\end{align*} and for any set $J\subseteq[n\rangle$, let
$$\mathcal B(n,J)\triangleq\{\bm e\in\mathbb R^n:
\mathsf{Supp}(\bm e)\subseteq J\}.$$

\begin{lem}\cite[Proposition 2]{{Wei24}}\label{lem-Hm-GM}
Suppose $\mathcal C$ is an $[n,k>0]$ linear code over $\mathbb R$
and $H$ is a parity check matrix of $\mathcal C$. Then for each
$m\in[1:\mathsf{d}(\mathcal C)]$, it holds that
$$\Gamma_{m}(\mathcal C)=\min \big\{\Delta\in\mathbb R^+:
H\bm e^\top\notin 2\mathcal S_H, \!~\forall\!~\bm e\in\mathcal
B_\Delta(n,m)\big\}.$$
\end{lem}

\begin{lem}\label{lem-Coef-Lth}
Suppose $\bm\alpha,\bm\beta\in\mathbb R^r\backslash\{\bm 0\}$ and
$0\leq\rho<1$ such that $\|\bm \alpha\|_2=1$ and
$\frac{|\langle\bm \alpha,\bm \beta\rangle|}{\|\beta\|_2}\leq
\rho$. Then for any $\bm u=a\bm \alpha+b\bm \beta$ with
$a,b\in\mathbb R$, we have $|a|\leq \frac{\|\bm
u\|_2}{\sqrt{1-\rho^2}}$.
\end{lem}
\begin{proof}
Since $\bm u=a\bm\alpha+b\bm\beta$ and $\|\bm \alpha\|_2=1$, then
\begin{align*}
\|\bm u\|_2^2&=\langle a\bm\alpha+b\bm\beta,a\bm\alpha+b\bm\beta\rangle\\
&=a^2+2ab\langle\bm\alpha,\bm\beta\rangle+b^2\|\bm\beta\|_2^2\\
&=\left(b\|\bm\beta\|_2+\frac{a\langle\bm\alpha,\bm\beta\rangle}{\|\bm\beta\|_2}\right)^2
+a^2\left(1-\left(\frac{\langle\bm\alpha,\bm\beta\rangle}{\|\bm\beta\|_2}\right)^2\right)\\
&\geq a^2(1-\rho^2)
\end{align*}
where the inequality holds because $0\leq\rho<1$ and
$\frac{|\langle\bm \alpha,\bm \beta\rangle|}{\|\beta\|_2}\leq
\rho$. Hence, we have $|a|\leq \frac{\|\bm
u\|_2}{\sqrt{1-\rho^2}}$.
\end{proof}

The following theorem gives an upper bound on the height profile
of codes whose parity check matrices have columns of unit
Euclidean norm.
\begin{thm}\label{thm-Con-Frmk}
Suppose $H$ is an $r\times n$ matrix over $\mathbb R$,
$m\in[1:r+1]$ and $0\leq \rho<1$ such that:
\begin{itemize} \item[(1)] Each column
of $H$ has unit Euclidean norm, i.e., $\|\bm h_j\|_2=1$ for each
$j\in[n\rangle$; \item[(2)] For any $j\in[n\rangle$, any
$J'\subseteq[n\rangle\backslash\{j\}$ of size $|J'|=m-1$ and any
$\bm u\in H*\mathcal B(n,J')$, it holds that $\frac{|\langle\bm
h_j,\bm u\rangle|}{\|\bm u\|_2}\leq \rho$.
\end{itemize}
Let $\mathcal C$ be the code that has $H$ as a parity check
matrix. Then we have $\Gamma_{m}(\mathcal C)\leq
\frac{2n}{\sqrt{1-\rho^2}}$.
\end{thm}
\begin{proof}
By definition, the set $H*\mathcal B(n,J')$ in condition (2)
actually coincides with the subspace of $\mathbb R^r$ spanned by
the vectors $\bm h_{j'}, j'\in J'$. As $0\leq \rho<1$, condition
(2) implies that any $m$ columns of $H$ are linearly independent,
and hence we have $\mathsf{d}(\mathcal C)>m$. Moreover, for any
$\bm u\in 2\mathcal S_H$, by the definition of $2\mathcal S_H$, we
have $\|\bm u\|_2\leq \sum_{\ell\in[n\rangle}2\|\bm
h_\ell\|_2=2n$.

Let $\Delta=\frac{2n}{\sqrt{1-\rho^2}}$. Further, for any
$J=\{j_1,j_2,\cdots,j_m\}\subseteq[n\rangle$, let $$\mathcal
B_\Delta(n,J)\triangleq\left\{\bm e\in\mathbb R^n:
\mathsf{Supp}(\bm e)\subseteq
J~\text{and}~\mathsf{Supp}_\Delta(\bm e)\neq\emptyset\right\}$$
and $$\overline{\mathcal B}_\Delta(n,J)\triangleq\left\{\bm
e\in\mathbb R^n: \mathsf{Supp}(\bm e)\subseteq
J~\text{and}~\mathsf{Supp}_\Delta(\bm e)=\emptyset\right\}.$$ We
first prove $\bm u\in H*\overline{\mathcal B}_\Delta(n,J)$ for any
$\bm u\in H*\mathcal B(n,J)$ with $\|\bm u\|_2\leq 2n$. In fact,
since $\bm u\in H*\mathcal B(n,J)$, we have $\bm u=a_{j_1}\bm
h_{j_1}+a_{j_2}\bm h_{j_2}+\cdots+a_{j_m}\bm h_{j_m}$ for some
$(a_{j_1},a_{j_2},\cdots,a_{j_m})\in\mathbb R^m$. For each
$\ell\in[1:m+1]$, by Conditions (1), (2) and by Lemma
\ref{lem-Coef-Lth} $($taking $\bm\alpha=\bm h_{j_\ell}$ and
$\bm\beta=\sum_{i\in[1:m+1]\backslash\{\ell\}}a_{j_i}\bm
h_{j_i})$, we can obtain $|a_{j_\ell}|\leq \frac{\|\bm
u\|_2}{\sqrt{1-\rho^2}}\leq \frac{2n}{\sqrt{1-\rho^2}}=\Delta$.
Therefore, we have $\bm u=H\bm a^\top$, where $\bm
a=(a_0,a_1,\cdots,a_{n-1})$ such that $a_j=0$ for
$j\in[n\rangle\backslash J$ and $|a_{j_\ell}|\leq\Delta$ for $j\in
J$. This implies that $\bm u\in H*\overline{\mathcal
B}_\Delta(n,J)$.

Now, we can prove $\Gamma_{m}(\mathcal C)\leq
\frac{2n}{\sqrt{1-\rho^2}}$. By Lemma \ref{lem-Hm-GM}, it suffices
to prove that $H\bm e^\top\notin 2\mathcal S_H$ for all $\bm
e\in\mathbb R^n$ such that $|\mathsf{Supp}(\bm e)|\leq
m~\text{and}~\mathsf{Supp}_\Delta(\bm e)\neq\emptyset$, where
$\Delta=\frac{2n}{\sqrt{1-\rho^2}}$. This can be proved by
contradiction. Let $J=\{j_1,j_2,\cdots,j_m\}\subseteq[n\rangle$ be
such that $\mathsf{Supp}(\bm e)\subseteq J$. Then $\bm e\in
\mathcal B_\Delta(n,J)$ and so $H\bm e^\top\in H*\mathcal
B_\Delta(n,J)\subseteq H*\mathcal B(n,J)$, where the inclusion
relationship comes from the simple fact that $\mathcal
B_\Delta(n,J)\subseteq \mathcal B(n,J)$. Suppose $\bm u=H\bm
e^\top\in 2\mathcal S_H$. Then $\|\bm u\|_2\leq 2n$ and by the
proven result we have $\bm u\in H*\overline{\mathcal
B}_\Delta(n,J)$. Therefore, \begin{align}\label{u-inCap}\bm u\in
\big(H*\mathcal B_\Delta(n,J)\big)\cap\big(H*\overline{\mathcal
B}_\Delta(n,J)\big).\end{align} On the other hand, by the
definition of $\mathcal B_\Delta(n,J)$ and $\overline{\mathcal
B}_\Delta(n,J)$, we can easily see that $\mathcal
B_\Delta(n,J)\cap\overline{\mathcal B}_\Delta(n,J)=\emptyset$.
Moreover, noticing that $\bm h_{j_1},\bm h_{j_2},\cdots,\bm
h_{j_m}$ are linearly independent, so $\big(H*\mathcal
B_\Delta(n,J)\big)\cap\big(H*\overline{\mathcal
B}_\Delta(n,J)\big)=\emptyset$, which contradicts to
\eqref{u-inCap}. Thus, we must have $\bm u=H\bm e^\top\notin
2\mathcal S_H$, which completes the proof.
\end{proof}

\begin{rem}
We can define
$$\mathsf{Coh}(\bm h_j,H_{J'})\triangleq
\min_{\bm u\in H*\mathcal B(n,J')}\frac{|\langle\bm h_j,\bm
u\rangle|}{\|\bm u\|_2}.$$ Geometrically, $\mathsf{Coh}(\bm
h_j,H_{J'})$ is the cosine of the principal angle between $\bm
h_j$ and the subspace spanned by $H_{J'}=\{\bm h_{j'}, j'\in
J'\}$. According to Theorem \ref{thm-Con-Frmk}, small value of
$\Gamma_{m}(\mathcal C)$ can be obtained by minimizing
$\rho=\max\{\mathsf{Coh}(\bm h_j,H_{J'}): j\in[n\rangle$,
$J'\subseteq[n\rangle\backslash\{j\}, |J'|=m-1\}$. In the special
case $m=2$, this reduces to minimizing $\rho=\max_{j\neq
j'\in[n\rangle}\{|\langle\bm h_j,\bm h_{j'}\rangle|\}$, which is
the classical spherical code (or coherence minimization) problem
studied in coding theory, discrete geometry, and harmonic
analysis, with applications to communication systems, frame
theory, and compressed sensing $($e.g., see \cite{Strohmer03,
Fickus18}$)$.
\end{rem}

In the rest of this section, we propose a decoder, denoted by
$\mathcal D_1$, that corrects a single error with respect to the
threshold pair $\big(\delta=1, \Delta=2n/\sqrt{1-\rho}\big)$ for
any code which has a parity check matrix $H$ as in Theorem
\ref{thm-Con-Frmk}. Specifically, we suppose $H$ is an $r\times n$
matrix over $\mathbb R$ such that $\|\bm h_j\|_2=1$ for each
column $\bm h_j$ of $H$ and $\rho=\max_{j\neq
j'\in[n\rangle}\{|\langle\bm h_j,\bm h_{j'}\rangle|\}<1$. Let
$\mathcal C$ be the $[n,k=n-r]$ linear code over $\mathbb R$ that
has $H$ as a parity check matrix. Then the decoder $\mathcal D_1$
of $\mathcal C$ can be defined as follows.

\textbf{The Decoder $\bm{\mathcal D_1}$}: Let
$\theta=\sqrt{\frac{1+\rho}{1-\rho}}n$ and $\Delta=\theta+n$. For
any $\bm y\in\mathbb R^n$, compute $\bm s=H\bm y^\top$ and compute
$\xi_j=\langle\bm s,\bm h_j\rangle=\bm h_j^\top\bm s$ for all
$j\in[n\rangle$. If $|\xi_j|\leq\theta$ for all $j\in[n\rangle$,
let $\mathcal D_1(\bm y)=\emptyset$; otherwise, choose a
$j_0\in[n\rangle$ such that $|\xi_{j_0}|\geq|\xi_{j}|$ for all
$j\in[n\rangle\backslash\{j_0\}$ and let $\mathcal D_1(\bm
y)=\{j_0\}$.

Before proving the correctness of the decoder $\mathcal D_1$, we
need the following two lemmas.

\begin{lem}\label{lem-vtcl-cdot}
Suppose $\bm\alpha,\bm\beta\in\mathbb R^r\backslash\{\bm 0\}$ such
that $\|\bm\alpha\|_2=\|\bm\beta\|_2=1$ and
$|\langle\bm\alpha,\bm\beta\rangle|=\bar{\rho}<1$. If
$\bm\gamma=a\bm\alpha+b\bm\beta$, where $a,b\in\mathbb R$, and
$|\langle\bm\alpha,\bm\gamma\rangle|=0$. Then
$|\langle\bm\beta,\bm\gamma\rangle|=\sqrt{1-\bar{\rho}^2}\|\bm
\gamma\|_2$.
\end{lem}
\begin{proof}
From the assumption $\bm\gamma=a\bm\alpha+b\bm\beta$, we can
obtain $\langle\bm\alpha,\bm\gamma\rangle=\langle\bm\alpha,
a\bm\alpha+b\bm\beta\rangle=a\langle\bm\alpha,\bm\alpha\rangle+b\langle\bm\alpha,
\bm\beta\rangle=a+b\bar{\rho}$. Since
$|\langle\bm\alpha,\bm\gamma\rangle|=0$, then we have
$a+b\bar{\rho}=0$, which implies
$$\bar{\rho}=-\frac{a}{b}.$$ Note that
$\|\gamma\|_2^2=\langle\bm\gamma,\bm\gamma\rangle=\langle
a\bm\alpha+b\bm\beta,
a\bm\alpha+b\bm\beta\rangle=a^2\langle\bm\alpha,\bm\alpha\rangle+2ab\langle\bm\alpha,
\bm\beta\rangle+b^2\langle\bm\beta,\bm\beta\rangle=a^2+b^2+2ab\bar{\rho}$,
where the last equality comes from the assumption that
$\|\bm\alpha\|_2=\|\bm\beta\|_2=1$. Then we can obtain
\begin{align*}(1-\bar{\rho}^2)\|\gamma\|_2^2&=(1-\bar{\rho}^2)(a^2+b^2+2ab\bar{\rho})\\
&=(1-\frac{a^2}{b^2})(a^2+b^2-2ab\frac{a}{b})\\&=\left(\frac{b^2-a^2}{b}\right)^2\end{align*}
where the second equality holds because $\bar{\rho}=-\frac{a}{b}$.
On the other hand, by the assumption,
$\langle\bm\beta,\bm\gamma\rangle^2=\langle\bm\beta,
a\bm\alpha+b\bm\beta\rangle^2=(a\langle\bm\beta,\bm\alpha\rangle+b\langle\bm\beta,
\bm\beta\rangle)^2=(a\bar{\rho}+b)^2=(-\frac{a^2}{b}+b)^2=(\frac{b^2-a^2}{b})^2$.
Therefore, we have
$\langle\bm\beta,\bm\gamma\rangle^2=(1-\bar{\rho}^2)\|\gamma\|_2^2$,
and so
$|\langle\bm\beta,\bm\gamma\rangle|=\sqrt{1-\bar{\rho}^2}\|\bm
\gamma\|_2$, which completes the proof.
\end{proof}

\begin{lem}\label{lem-exp-dvn}
Suppose $\bm y=\bm c+\bm\varepsilon+\bm e$ such that $\bm
c\in\mathcal C$, $\bm\varepsilon\in\mathcal Q(n,1)$ and
$\mathsf{Supp}(\bm e)\subseteq\{j_0\}$. Let $\theta$, $\Delta$ and
$\xi_{j}$, $j\in[n\rangle$, be defined as in the decoder
$\bm{\mathcal D_1}$. The following statements hold.
\begin{itemize}
 \item[1)] If $|e_{j_0}|>\Delta$, then $|\xi_{j_0}|>\theta$.
 \item[2)] If $|\xi_{j_0}|>\theta$, then
 $|\xi_{j_0}|>|\xi_{j}|$ for all $j\in[n\rangle\backslash\{j_0\}$.
 \item[3)] If $|\xi_{j_0}|\leq\theta$, then $|\xi_{j}|\leq \theta$
 for all $j\in[n\rangle\backslash\{j_0\}$.
\end{itemize}
\end{lem}
\begin{proof}
Since $\bm c\in\mathcal C$, we have $H\bm c^\top=0$. Then by the
definition of $\bm s$ and by the assumption that
$\mathsf{Supp}(\bm e)\subseteq\{j_0\}$, we can obtain $\bm s=H\bm
y^\top=H(\bm c^\top+\bm\varepsilon^\top+\bm
e^\top)=H\bm\varepsilon^\top+H\bm e^\top=e_{j_0}\bm
h_{j_0}+\sum_{\ell=0}^{n-1}\varepsilon_\ell\bm h_\ell$. Therefore,
\begin{align}\label{eq-s-cmpt} \bm s=e_{j_0}\bm
h_{j_0}+\sum_{\ell=0}^{n-1}\varepsilon_\ell\bm h_\ell
\end{align} and
for each $j\in[n\rangle$, we have
\begin{align}\label{eq-xi-cmpt}
\xi_{j}=\langle\bm s, \bm h_j\rangle=e_{j_0}\langle\bm h_j,\bm
h_{j_0}\rangle+\sum_{\ell=0}^{n-1}\varepsilon_\ell\langle\bm
h_j,\bm h_\ell\rangle.
\end{align}

We first prove the following claim.

\emph{Claim 1}: $|\xi_{j}|\leq
|\xi_{j_0}\rho_{j,j_0}|+\sqrt{1-\rho_{j,j_0}^2}n$ for all
$j\in[n\rangle\backslash\{j_0\}$, where $\rho_{j,j_0}=\langle\bm
h_j,\bm h_{j_0}\rangle$.

To prove Claim 1, the key is to write $\bm s=\xi_{j_0}\bm
h_{j_0}+\bm u'+\bm u''$ such that for each
$j\in[n]\backslash\{j_0\}$, $\langle\bm u',\bm h_j\rangle\leq
\sqrt{1-\rho_{j,j_0}^2}n$ and $\langle\bm u'',\bm h_j\rangle=0$.
This can be achieved by two steps. In the first step, let $\bm
v=\sum_{\ell=0}^{n-1}\varepsilon_\ell\bm h_\ell$ and $\bm u=\bm
v-\langle\bm v,\bm h_{j_0}\rangle\bm h_{j_0}$. Then $\|\bm
v\|_2\leq \sum_{\ell=0}^{n-1}|\varepsilon_\ell|\cdot\|\bm
h_\ell\|_2\leq n$ and $\bm v=\bm u+\langle\bm v,\bm
h_{j_0}\rangle\bm h_{j_0}$. From \eqref{eq-s-cmpt}, we have $\bm
s=\bm v+e_{j_0}\bm h_{j_0}=(e_{j_0}+\langle\bm v,\bm
h_{j_0}\rangle)\bm h_{j_0}+\bm u$. By the definition of $\bm u$
and $\bm v$, we can verify $\langle\bm u,\bm h_{j_0}\rangle=0$,
$\|\bm u\|_2\leq n$ and $\xi_{j_0}=e_{j_0}+\langle\bm v,\bm
h_{j_0}\rangle$ as follows:
\begin{itemize} \item[i)] $\langle\bm u,\bm h_{j_0}\rangle=\langle\bm
v-\langle\bm v,\bm h_{j_0}\rangle\bm h_{j_0},\bm
h_{j_0}\rangle=\langle\bm v,\bm h_{j_0}\rangle-\langle\bm v,\bm
h_{j_0}\rangle\langle\bm h_{j_0},\bm h_{j_0}\rangle=\langle\bm
v,\bm h_{j_0}\rangle-\langle\bm v,\bm h_{j_0}\rangle=0$, where the
third equality holds because $\|\bm h_{j_0}\|_2=1$; \item[ii)]
$\|\bm u\|_2^2=\langle\bm v-\langle\bm v,\bm h_{j_0}\rangle\bm
h_{j_0},\bm v-\langle\bm v,\bm h_{j_0}\rangle\bm
h_{j_0}\rangle=\langle\bm v,\bm v\rangle-2\langle\bm v,\bm
h_{j_0}\rangle^2+\langle\bm v,\bm h_{j_0}\rangle^2=\|\bm
v\|_2^2-\langle\bm v,\bm h_{j_0}\rangle^2\leq \|\bm v\|_2^2\leq
n^2~($noticing that $\|\bm v\|_2\leq n)$, hence $\|\bm u\|_2\leq
n$; \item[iii)] $\xi_{j_0}=\langle\bm s,\bm
h_{j_0}\rangle=\langle(e_{j_0}+\langle\bm v,\bm h_{j_0}\rangle)\bm
h_{j_0}+\bm u,\bm h_{j_0}\rangle=(e_{j_0}+\langle\bm v,\bm
h_{j_0}\rangle)\langle\bm h_{j_0},\bm h_{j_0}\rangle+\langle\bm
u,\bm h_{j_0}\rangle=e_{j_0}+\langle\bm v,\bm h_{j_0}\rangle$,
where the last equality holds because $\|\bm h_{j_0}\|_2=1$ and we
have shown in i) that $\langle\bm u,\bm
h_{j_0}\rangle=0$.\end{itemize} Hence, we can write $\bm
s=\xi_{j_0}\bm h_{j_0}+\bm u$ such that $\langle\bm u,\bm
h_{j_0}\rangle=0$ and $\|\bm u\|_2\leq n$. In the second step, for
each $j\in[n]\backslash\{j_0\}$, let $\bm u'=\frac{\langle\bm
u,\bm h_j\rangle-\rho_{j,j_0}\langle\bm u,\bm
h_{j_0}\rangle}{1-\rho_{j,j_0}^2}\bm h_{j}+\frac{\langle\bm u,\bm
h_{j_0}\rangle-\rho_{j,j_0}\langle\bm u,\bm
h_{j}\rangle}{1-\rho_{j,j_0}^2}\bm h_{j_0}$ and $\bm u''=\bm u-\bm
u'$. Then $\bm u=\bm u'+\bm u''$ and so $\bm s=\xi_{j_0}\bm
h_{j_0}+\bm u'+\bm u''$. For each $j\in[n]\backslash\{j_0\}$, we
can verify that $\langle\bm u'',\bm h_j\rangle=0$ and $\langle\bm
u',\bm h_j\rangle\leq \sqrt{1-\rho_{j,j_0}^2}n$ as follows:
\begin{itemize} \item[iv)] By the definition of $\bm u'$
and $\bm u''$, and by the assumption $\|\bm h_j\|_2=1$, we can
obtain
\begin{align*}\langle\bm u',\bm h_j\rangle&=\frac{\langle\bm u,\bm
h_j\rangle-\rho_{j,j_0}\langle\bm u,\bm
h_{j_0}\rangle}{1-\rho_{j,j_0}^2}\\&~~~~+\frac{\langle\bm u,\bm
h_{j_0}\rangle-\rho_{j,j_0}\langle\bm u,\bm
h_{j}\rangle}{1-\rho_{j,j_0}^2}\rho_{j,j_0}\\
&=\frac{1}{1-\rho_{j,j_0}^2}(1-\rho_{j,j_0}^2)\langle\bm u,\bm
h_{j}\rangle\\&=\langle\bm u,\bm h_{j}\rangle.\end{align*} Since
$\bm u''=\bm u-\bm u'$, we have $\langle\bm u'',\bm
h_{j}\rangle=\langle\bm u-\bm u',\bm h_{j}\rangle=\langle\bm u,\bm
h_{j}\rangle-\langle\bm u',\bm h_{j}\rangle=0$. \item[v)] Similar
to iv), we have
\begin{align*}\langle\bm u',\bm h_{j_0}\rangle&=\frac{\langle\bm
u,\bm h_j\rangle-\rho_{j,j_0}\langle\bm u,\bm
h_{j_0}\rangle}{1-\rho_{j,j_0}^2}\rho_{j,j_0}\\&~~~~+\frac{\langle\bm
u,\bm h_{j_0}\rangle-\rho_{j,j_0}\langle\bm u,\bm
h_{j}\rangle}{1-\rho_{j,j_0}^2}\\
&=\frac{1}{1-\rho_{j,j_0}^2}(1-\rho_{j,j_0}^2)\langle\bm u,\bm
h_{j_0}\rangle\\&=\langle\bm u,\bm
h_{j_0}\rangle\\&=0,\end{align*} so $\langle\bm u'',\bm
h_{j_0}\rangle=\langle\bm u-\bm u',\bm h_{j_0}\rangle=\langle\bm
u,\bm h_{j}\rangle-\langle\bm u',\bm h_{j}\rangle=0$. Noticing
that $\bm u'$ is a linear combination of $\bm h_{j}$ and $\bm
h_{j_0}$, and we have shown that $\langle\bm u'',\bm
h_{j}\rangle=0$ and $\langle\bm u'',\bm h_{j_0}\rangle=0$, then we
can obtain $\langle\bm u'',\bm u'\rangle=0$, and so we have $\|\bm
u\|_2^2=\langle\bm u'+\bm u'',\bm u'+\bm u''\rangle=\langle\bm
u',\bm u'\rangle+\langle\bm u'',\bm u''\rangle=\|\bm
u'\|_2^2+\|\bm u''\|_2^2$, which implies that $\|\bm
u'\|_2\leq\|\bm u\|_2\leq n$. Since $\langle\bm u',\bm
h_{j_0}\rangle=0$, by Lemma \ref{lem-vtcl-cdot}, we can obtain
$\langle\bm u',\bm h_j\rangle=\sqrt{1-\rho_{j,j_0}^2}\|\bm
u'\|_2\leq\sqrt{1-\rho_{j,j_0}^2}n$.\end{itemize} By the above
discussions, we can obtain $\bm s=\xi_{j_0}\bm h_{j_0}+\bm u'+\bm
u''$ such that for each $j\in[n]\backslash\{j_0\}$, $\langle\bm
u',\bm h_j\rangle\leq \sqrt{1-\rho_{j,j_0}^2}n$ and $\langle\bm
u'',\bm h_j\rangle=0$. From this and by the definition of
$\xi_{j}$, we can obtain $|\xi_{j}|=|\langle\bm s,\bm
h_j\rangle|=|\langle\xi_{j_0}\bm h_{j_0}+\bm u'+\bm u'',\bm
h_j\rangle|\leq |\xi_{j_0}\langle\bm h_{j_0},\bm
h_{j}\rangle|+|\langle\bm u',\bm h_j\rangle|\leq
|\xi_{j_0}\rho_{j,j_0}|+\sqrt{1-\rho_{j,j_0}^2}n$, which proves
Claim 1.

Now, we can prove statements 1), 2) and 3) as follows.

1) Note that by \eqref{eq-xi-cmpt}, we have
\begin{align*}|\xi_{j_0}|&=|e_{j_0}+\sum_{\ell=0}^{n-1}\varepsilon_\ell\langle\bm
h_{j_0},\bm h_\ell\rangle|\\&\geq
|e_{j_0}|-\sum_{\ell=0}^{n-1}|\varepsilon_\ell\langle\bm
h_{j_0},\bm h_\ell\rangle|\\&\geq |e_{j_0}|-n.\end{align*} Hence,
if $|e_{j_0}|>\Delta=\theta+n$, then $|\xi_{j_0}|\geq
|e_{j_0}|-n>\theta$.

2) Since $0\leq|\rho_{j,j_0}|\leq\rho<1$, we have
$\sqrt{\frac{1+|\rho_{j,j_0}|}{1-|\rho_{j,j_0}|}}n\leq\sqrt{\frac{1+\rho}{1-\rho}}n=\theta$.
Moreover, since by assumption
$|\xi_{j_0}|>\theta=\sqrt{\frac{1+\rho}{1-\rho}}n$, then we can
obtain
$\sqrt{\frac{1+|\rho_{j,j_0}|}{1-|\rho_{j,j_0}|}}n<|\xi_{j_0}|$,
which implies
$$|\xi_{j_0}|\sqrt{1-|\rho_{j,j_0}|}-\sqrt{1+|\rho_{j,j_0}|}n>0.$$
By Claim 1, we have $|\xi_{j}|\leq
|\xi_{j_0}\rho_{j,j_0}|+\sqrt{1-\rho_{j,j_0}^2}n$ for all
$j\in[n\rangle\backslash\{j_0\}$, then
\begin{align*}|\xi_{j_0}|-|\xi_{j}|&\geq|\xi_{j_0}|
-\left(|\xi_{j_0}\rho_{j,j_0}|+\sqrt{1-\rho_{j,j_0}^2}n\right)\\
&=|\xi_{j_0}|(1-|\rho_{j,j_0}|)-\sqrt{1-\rho_{j,j_0}^2}n\\
&=\sqrt{1-|\rho_{j,j_0}|}\Big(|\xi_{j_0}|\sqrt{1-|\rho_{j,j_0}|}
-\sqrt{1+|\rho_{j,j_0}|}n\Big)\\&>0\end{align*} which implies
$|\xi_{j_0}|>|\xi_{j}|$.

3) Since by assumption
$|\xi_{j_0}|\leq\theta=\sqrt{\frac{1+\rho}{1-\rho}}n$, and by
Claim 1, $|\xi_{j}|\leq
|\xi_{j_0}\rho_{j,j_0}|+\sqrt{1-\rho_{j,j_0}^2}n$ for all
$j\in[n\rangle\backslash\{j_0\}$, then
\begin{align}\label{eq1-exp-dvn}|\xi_{j}|&\leq
|\xi_{j_0}\rho_{j,j_0}|+\sqrt{1-\rho_{j,j_0}^2}n\nonumber\\
&\leq\theta|\rho_{j,j_0}|+\sqrt{1-\rho_{j,j_0}^2}n.\end{align}
Further, since $0\leq|\rho_{j,j_0}|\leq\rho<1$, then
$\sqrt{\frac{1+|\rho_{j,j_0}|}{1-|\rho_{j,j_0}|}}n\leq\sqrt{\frac{1+\rho}{1-\rho}}n=\theta$,
which implies
\begin{align}\label{eq2-exp-dvn}\theta\sqrt{1-|\rho_{j,j_0}|}\geq\sqrt{1+|\rho_{j,j_0}|}n.\end{align}
Therefore, for all $j\in[n\rangle\backslash\{j_0\}$, we can obtain
\begin{align*}\theta-|\xi_{j}|&\geq
\theta-\Big(\theta|\rho_{j,j_0}|+\sqrt{1-\rho_{j,j_0}^2}n\Big)\\
&=\theta(1-|\rho_{j,j_0}|)-\sqrt{1-\rho_{j,j_0}^2}n\\
&=\sqrt{1-|\rho_{j,j_0}|}\Big(\theta\sqrt{1-|\rho_{j,j_0}|}-\sqrt{1+|\rho_{j,j_0}|}n\Big)\\&\geq
0\end{align*} where the first inequality comes from
\eqref{eq1-exp-dvn} and the second inequality comes from
\eqref{eq2-exp-dvn}. Thus, we have $|\xi_{j}|\leq \theta$.
\end{proof}

The correctness of the decoder $\mathcal D_1$ is given by the
following theorem.
\begin{thm}\label{thm-crt-1-err}
The decoder $\mathcal D_1$ of $\mathcal C$ can correct a single
error with respect to the threshold pair $\big(\delta=1, \Delta=
\frac{\sqrt{1+\rho}+\sqrt{1-\rho}}{\sqrt{1-\rho}}n\big)$.
\end{thm}
\begin{proof}
According to Definition \ref{def-anl-dec}, we need to prove that
$\mathsf{Supp}_\Delta(\bm e)\subseteq\mathcal D(\bm y)\subseteq
\mathsf{Supp}(\bm e)$ for all $\bm y=\bm c+\bm\varepsilon+\bm e$
such that $\bm c\in\mathcal C$, $\bm\varepsilon\in\mathcal
Q(n,\delta)$ and $\bm e\in\mathcal B(n,1)$. This can be easily
proved by Lemma \ref{lem-exp-dvn}. In fact, suppose
$\mathsf{Supp}_\Delta(\bm e)\subseteq\{j_0\}$. Then we have the
following arguments.
\begin{itemize}
 \item If $|e_{j_0}|>\Delta$, then $\mathsf{Supp}_\Delta(\bm e)=\{j_0\}$. By 1) of Lemma
 \ref{lem-exp-dvn}, we have $|e_{j_0}|>\Delta$. Further, by 2) of Lemma
 \ref{lem-exp-dvn} and by the definition of $\mathcal D_1$, we
 have $\mathcal D_1(\bm y)=\{j_0\}$. Hence, $\mathcal D_1(\bm y)=\mathsf{Supp}_\Delta(\bm e)=
 \mathsf{Supp}(\bm e)$.
 \item If $|e_{j_0}|\leq\Delta$ and $|\xi_{j_0}|>\theta$, then $\mathsf{Supp}_\Delta(\bm
 e)=\emptyset$. By 2) of Lemma
 \ref{lem-exp-dvn} and by the definition of $\mathcal D_1$, we
 have $\mathcal D_1(\bm y)=\{j_0\}$. Hence, $\mathsf{Supp}_\Delta(\bm
 e)\subseteq\mathcal D_1(\bm y)=\mathsf{Supp}(\bm e)$.
 \item If $|e_{j_0}|\leq\Delta$ and $|\xi_{j_0}|<\theta$, then $\mathsf{Supp}_\Delta(\bm
 e)=\emptyset$. By 3) of Lemma
 \ref{lem-exp-dvn} and by the definition of $\mathcal D_1$, we
 have $\mathcal D_1(\bm y)=\emptyset$ and so $\mathsf{Supp}_\Delta(\bm
 e)=\mathcal D_1(\bm y)\subseteq\mathsf{Supp}(\bm e)$.
\end{itemize}
Thus, we always have $\mathsf{Supp}_\Delta(\bm e)\subseteq\mathcal
D(\bm y)\subseteq \mathsf{Supp}(\bm e)$. By Definition
\ref{def-anl-dec}, $\mathcal D_1$ can correct a single error with
respect to the threshold pair $\big(\delta=1, \Delta=
\frac{\sqrt{1+\rho}+\sqrt{1-\rho}}{\sqrt{1-\rho}}n\big)$.
\end{proof}

\begin{rem}
We can verify that
$\frac{2n}{\sqrt{1-\rho^2}}<\Delta\leq\frac{2n}{\sqrt{1-\rho}}
<\sqrt{2}\cdot\frac{2n}{\sqrt{1-\rho^2}}$. In fact, by definition,
we have
$\Delta=\theta+n=\left(\sqrt{\frac{1+\rho}{1-\rho}}+1\right)n=
\frac{\sqrt{1+\rho}+\sqrt{1-\rho}}{\sqrt{1-\rho}}n$ and
$0\leq\rho<1$. It is easy to see that
$\frac{2}{\sqrt{1-\rho^2}}<\frac{\sqrt{1+\rho}+\sqrt{1-\rho}}{\sqrt{1-\rho}}
\leq\frac{2}{\sqrt{1-\rho}}$. Hence,
$\frac{2n}{\sqrt{1-\rho^2}}<\Delta\leq\frac{2n}{\sqrt{1-\rho}}
=\sqrt{1+\rho}\frac{2n}{\sqrt{1-\rho^2}}<\sqrt{2}\cdot\frac{2n}{\sqrt{1-\rho^2}}$.
\end{rem}

\section{Single error correcting codes with redundancy three}

In this section, we construct a family of $[n,k=n-3]$ linear codes
$\mathcal C$ for any $n>3$ with $\Gamma_{2}(\mathcal C)\leq
2\sqrt{2}n/\sin(\pi/\sqrt{n-1})$ and the decoder $\mathcal D_1$ is
applicable to this family of codes. Specifically, we have the
following construction.

\textbf{Construction 1}: Let $t>3$ be an integer. Let
$\Omega=\bigcup_{i=0}^{t}\Omega_i$ such that
$\Omega_0=\{(0,0,1)^\top\},$
\begin{align*}\Omega_i&=\left\{(\sin\phi_i\cos\theta_{i,j},
\sin\phi_i\sin\theta_{i,j}, \cos\phi_i)^\top:
~\phi_i=\frac{\pi}{2t}i,\right.
\\&~~~~~~\left.\theta_{i,j}=\frac{\pi}{2i}j,~ j\in[4i\rangle\right\},
~~i\in[1:t]\end{align*} and
\begin{align*}\Omega_i&=\left\{(\sin\phi_i\cos\theta_{i,j},
\sin\phi_i\sin\theta_{i,j}, \cos\phi_i)^\top:~
\phi_i=\frac{\pi}{2t}i,\right.
\\&~~~~~~\left.\theta_{i,j}=\frac{\pi}{2i}j,~ j\in[2i\rangle\right\},
~~i=t.\end{align*} Let $H$ be the $3\times n$ matrix whose columns
are all vectors in $\Omega$ and $\mathcal C$ be the code that has
$H$ as a parity check matrix.

Clearly, $\Omega_t=\left\{(\cos\theta_{t,j}, \sin\theta_{t,j},
0)^\top: \theta_{t,j}=\frac{\pi}{2t}j,~ j\in[2t\rangle]\right\}$
and $n=|\Omega|=1+4\sum_{i=1}^{t-1}i+2t=2t^2+1$. To estimate
$\Gamma_{2}(\mathcal C)$, we need the following two lemmas.

\begin{lem}\label{lem-agl-dlong} Suppose $\ell\geq 1$ and
$0\leq x\leq\frac{\pi}{2}$. Then we have
$\sin\frac{\pi}{2\ell}\sin x\geq\sin\frac{x}{\ell}$.
\end{lem}
\begin{proof}
It suffices to prove $\sin x\sin y-\sin(\frac{2}{\pi}xy)\geq 0$
for all $x,y\in[0,\frac{\pi}{2}]$. In fact, the conclusion of this
lemma can be obtained from this result by letting
$y=\frac{\pi}{2\ell}$.

Clearly, if $y=0$ or $y=\frac{\pi}{2}$, then for all
$x\in[0,\frac{\pi}{2}]$, we have $\sin x\sin
y=\sin(\frac{2}{\pi}xy)$. So, we only need to consider
$y\in(0,\frac{\pi}{2})$. For each fixed $y$, denote
$$f(x)=\sin x\sin
y-\sin\Big(\frac{2}{\pi}xy\Big),~~x\in\Big[0,\frac{\pi}{2}\Big].$$
Then it suffices to prove $f(x)\geq 0$ for all
$x\in[0,\frac{\pi}{2}]$. To prove this, we first prove the
following two claims.

\emph{Claim 2}: $1>\sin y>\frac{2}{\pi}y>0$ for all
$y\in(0,\frac{\pi}{2})$.
\begin{proof}[Proof of Claim 2]
It is easy to see that $1>\sin y$ and $\frac{2}{\pi}y>0$
for all $y\in(0,\frac{\pi}{2})$. To prove $\sin y>\frac{2}{\pi}y$,
consider the function $h(y)=\sin y-\frac{2}{\pi}y$ with
$y\in[0,\frac{\pi}{2}]$. We have $h(0)=h(\frac{\pi}{2})=0$ and
$h'(y)=\cos y-\frac{2}{\pi}$. Clearly, $h'(y)>0$ for
$y\in[0,\arcsin\frac{2}{\pi})$ and $h'(y)<0$ for
$y\in(\arcsin\frac{2}{\pi},\frac{\pi}{2}]$, so
$h(y)>h(0)=h(\frac{\pi}{2})=0$ for all $y\in(0,\frac{\pi}{2})$,
which implies that $\sin y>\frac{2}{\pi}y$.\end{proof}

\emph{Claim 3}: There exists an $x_0\in(0,\frac{\pi}{2})$ such
that $f(x)$ is increasing in $[0,x_0]$ and $f(x)$ is decreasing in
$[x_0,\frac{\pi}{2}]$.
\begin{proof}[Proof of Claim 3]
To prove Claim 3, we note that
$$f'(x)=\sin y\cos
x-\frac{2}{\pi}y\cos\Big(\frac{2}{\pi}xy\Big)$$ and for
$x\in(0,\frac{\pi}{2})$, we have
\begin{align*}
f''(x)&=-\sin y\sin
x+\left(\frac{2}{\pi}y\right)^2\sin\Big(\frac{2}{\pi}xy\Big)\\
&<-\sin y\sin x+\left(\frac{2}{\pi}y\right)^2\sin(x)\\
&=-\sin x\left(\sin y-\left(\frac{2}{\pi}y\right)^2\right)\\&<0
\end{align*}
where the first inequality holds because $y\in(0,\frac{\pi}{2})$
and the last inequality holds because by Claim 2, we have $1>\sin
y>\frac{2}{\pi}y>0$ for $y\in(0,\frac{\pi}{2})$. Therefore,
$f'(x)$ is decreasing in $[0,\frac{\pi}{2}]$. Clearly, for any
fixed $y\in\Big(0,\frac{\pi}{2}\Big)$, we have
$$f'\Big(\frac{\pi}{2}\Big)=-\frac{2}{\pi}y\cos y<0$$ and by Claim 2, we have
$$f'(0)=\sin y-\frac{2}{\pi}y>0.$$ Then there exists an
$x_0\in(0,\frac{\pi}{2})$ such that $f'(x_0)=0$, $f'(x_0)>0$ for
$x\in(0,x_0)$ and $f'(x_0)<0$ for $x\in(x_0,\frac{\pi}{2})$, which
implies that $f(x)$ is increasing in $[0,x_0]$ and $f(x)$ is
decreasing in $[x_0,\frac{\pi}{2}]$. 
\end{proof}

Note that by the definition of $f(x)$, we can obtain
$f(0)=f(\frac{\pi}{2})=0$. Then by Claim 3, $f(x)\geq 0$ for all
$x\in[x_0,\frac{\pi}{2}]$.
\end{proof}

\begin{lem}\label{lem-Cnstrn-bnd}
Let $\Omega$ be obtained from Construction 1. It holds that
$|\langle\bm u, \bm u'\rangle|\leq \cos\frac{\pi}{2t}$ for all
distinct $\bm u,\bm u'\in\Omega$.
\end{lem}
\begin{proof}
To prove the conclusion of this lemma, we need to consider the
following four cases.

Case 1: $\bm u\in\Omega_0$ and $\bm u'\in\cup_{i=1}^t\Omega_i$. By
the construction, it is easy to see that $\langle\bm u, \bm
u'\rangle=\cos\phi_i$, where $\phi_i=\frac{\pi}{2t}i$ and $1\leq
i\leq t$. So, $|\langle\bm u, \bm
u'\rangle|=|\cos(\frac{\pi}{2t}i)|\leq \cos\frac{\pi}{2t}$, where
the inequality holds because $\frac{\pi}{2t}\leq \phi_i
=\frac{\pi}{2t}i\leq \frac{\pi}{2}$ for $1\leq i\leq t$.

Case 2: $\bm u\in\Omega_i$ and $\bm u'\in\Omega_{i'}$ such that
$1\leq i<i'\leq t$. By the construction, $\bm
u=(\sin\phi_i\cos\theta_{i,j}, \sin\phi_i\sin\theta_{i,j},
\cos\phi_i)^\top$ and $\bm u'=(\sin\phi_{i'}\cos\theta_{i',j'},
\sin\phi_{i'}\sin\theta_{i',j'}, \cos\phi_{i'})^\top$ such that
$j\in[4i\rangle$ and $j'\in[4i'\rangle~($or $j'\in[2i'\rangle$ if
$i'=t)$, so we have
\begin{align*}\langle\bm u, \bm
u'\rangle&=\sin\phi_i\sin\phi_{i'}\cos\theta_{i,j}\cos\theta_{i',j'}\\
&~~~+
\sin\phi_i\sin\phi_{i'}\sin\theta_{i,j}\sin\theta_{i',j'}+\cos\phi_i\cos\phi_{i'}\\
&=\sin\phi_i\sin\phi_{i'}\cos(\theta_{i,j}-\theta_{i',j'})+
\cos\phi_i\cos\phi_{i'}\end{align*} 
which implies that $|\langle\bm u, \bm
u'\rangle|\leq|\sin\phi_i\sin\phi_{i'}|+|\cos\phi_i\cos\phi_{i'}|$.
Also by the construction, we can obtain
$\frac{\pi}{2t}\leq\phi_i<\phi_{i'}\leq\frac{\pi}{2}$ and
$\frac{\pi}{2t}\leq\phi_{i'}-\phi_{i}<\frac{\pi}{2}$. Therefore,
we have $|\langle\bm u, \bm
u'\rangle|\leq|\sin\phi_i\sin\phi_{i'}|+
|\cos\phi_i\cos\phi_{i'}|=\sin\phi_i\sin\phi_{i'}+
\cos\phi_i\cos\phi_{i'}=\cos(\phi_{i'}-\phi_{i})\leq\cos\frac{\pi}{2t}$.

Case 3: $\bm u, \bm u'\in\Omega_t$. By the construction, we have
$\bm u=\big(\cos(\frac{\pi}{2t}j), \sin(\frac{\pi}{2t}j),
0\big)^\top$ and $\bm u'=\big(\cos(\frac{\pi}{2t}j'),
\sin(\frac{\pi}{2t}j'), 0\big)^\top$ such that $0\leq j<j'\leq
2t-1$, so
\begin{align*}|\langle\bm u, \bm
u'\rangle|&=\Big|\cos\Big(\frac{\pi}{2t}j\Big)\cos\Big(\frac{\pi}{2t}j'\Big)+
\sin\Big(\frac{\pi}{2t}j\Big)\sin\Big(\frac{\pi}{2t}j'\Big)\Big|\\&=\Big|\cos\Big(\frac{\pi}{2t}(j'-j)\Big)\Big|\\
&\leq\cos\frac{\pi}{2t}\end{align*} 
where the inequality holds because
$\frac{\pi}{2t}\leq \frac{\pi}{2t}(j'-j)\leq \pi-\frac{\pi}{2t}$
for $0\leq j<j'\leq 2t-1$.

Case 4: $\bm u,\bm u'\in\Omega_i$ for some $i\in[1:t]$. By the
construction, we can let $\bm u=(\sin\phi_i\cos\theta_{i,j},
\sin\phi_i\sin\theta_{i,j}, \cos\phi_i)^\top$ and $\bm
u'=(\sin\phi_{i}\cos\theta_{i,j'}, \sin\phi_{i}\sin\theta_{i,j'},
\cos\phi_{i})^\top$ such that $0\leq j<j'\leq 4i-1$. Therefore,
\begin{align}\label{dot-u-up}
\langle\bm u, \bm u'\rangle&=\sin^2\phi_i\cos\theta_{i,j}\cos\theta_{i,j'}+
\sin^2\phi_i\sin\theta_{i,j}\sin\theta_{i,j'}\nonumber\\&~~~+\cos^2\phi_i\nonumber\\
&=\sin^2\phi_i\cos(\theta_{i,j'}-\theta_{i,j})+
\cos^2\phi_i.\end{align} Also by the construction, we have
$\frac{\pi}{2i}\leq
\theta_{i,j'}-\theta_{i,j}\leq2\pi-\frac{\pi}{2i}$, so
$-\sin^2\phi_i+\cos^2\phi_i\leq\sin^2\phi_i\cos(\theta_{i,j'}-\theta_{i,j})+
\cos^2\phi_i\leq \sin^2\phi_i\cos(\frac{\pi}{2i})+ \cos^2\phi_i$,
or equivalently,
\begin{align}\label{eq-lft}
\cos(2\phi_i)\leq\sin^2\phi_i\cos(\theta_{i,j'}-\theta_{i,j})+
\cos^2\phi_i\end{align} and
\begin{align}\label{eq-rght}\sin^2\phi_i\cos(\theta_{i,j'}-\theta_{i,j})+
\cos^2\phi_i\leq \sin^2\phi_i\cos\Big(\frac{\pi}{2i}\Big)+
\cos^2\phi_i.\end{align} Noticing that $\frac{\pi}{t}\leq
2\phi_i\leq\pi-\frac{\pi}{t}~\big($because $i\in[1:t]$ and by
Construction 1, $\phi_i=\frac{\pi}{2t}i\big)$, then
$-\cos\frac{\pi}{2t}\leq-\cos\frac{\pi}{t}\leq \cos2\phi_i$, so by
\eqref{eq-lft}, we can obtain
\begin{align}\label{eq-abs-lft}
-\cos\frac{\pi}{2t}\leq\sin^2\phi_i\cos(\theta_{i,j'}-\theta_{i,j})+
\cos^2\phi_i.\end{align} On the other hand, we can prove
$\sin^2\phi_i\cos(\theta_{i,j'}-\theta_{i,j})+
\cos^2\phi_i\leq\cos\frac{\pi}{2t}$ as follows. Let
$\frac{\pi}{2t}i=x$ and $2i=\ell$, then
$\frac{\pi}{4t}=\frac{x}{\ell}$ and by Lemma \ref{lem-agl-dlong},
we can obtain
\begin{align}\label{eq-frm-lem}
\sin\Big(\frac{\pi}{2t}i\Big)\sin\frac{\pi}{4i}\geq
\sin\frac{\pi}{4t}.\end{align} Therefore,
\begin{align}\label{eq-usn-lem}\sin^2\phi_i\cos\frac{\pi}{2i}+
\cos^2\phi_i&=\sin^2\phi_i\cos\frac{\pi}{2i}+
1-\sin^2\phi_i\nonumber\\&=1-\sin^2\phi_i\left(1-\cos\frac{\pi}{2i}\right)\nonumber\\
&=1-2\sin^2\phi_i\sin^2\frac{\pi}{4i}\nonumber\\
&=1-2\sin^2\Big(\frac{\pi}{2t}i\Big)\sin^2\frac{\pi}{4i}\nonumber\\&\leq
1-2\sin^2\frac{\pi}{4t}\nonumber\\&=\cos\frac{\pi}{2t}\end{align}
where the inequality comes from \eqref{eq-frm-lem}. Combining
\eqref{eq-usn-lem} and \eqref{eq-rght}, we can obtain
\begin{align}\label{eq-abs-rght}\sin^2\phi_i\cos(\theta_{i,j'}-\theta_{i,j})+ \cos^2\phi_i\leq
\cos\frac{\pi}{2t}.\end{align} From \eqref{dot-u-up},
\eqref{eq-abs-lft} and \eqref{eq-abs-rght}, we can obtain
$|\langle\bm u, \bm
u'\rangle|=|\sin^2\phi_i\cos(\theta_{i,j'}-\theta_{i,j})+
\cos^2\phi_i|\leq \cos\frac{\pi}{2t}$.

By the above discussions, we proved $|\langle\bm u, \bm
u'\rangle|\leq \cos\frac{\pi}{2t}$ for all distinct $\bm u,\bm
u'\in\Omega$.
\end{proof}

About the height profile and decoding of the codes obtained from
Construction 1, we have the following theorem.

\begin{thm}\label{Cntrn-rdncy-3}
Let $H$ and $\mathcal C$ be obtained from Construction 1. Then
$$\Gamma_{2}(\mathcal C)\leq
\frac{2n}{\sin\frac{\pi}{\sqrt{2(n-1)}}}.$$ Moreover, the decoder
$\mathcal D_1$ for $\mathcal C$ can correct a single error with
respect to the threshold pair $(\delta, \Delta)$ such that
$\delta=1$ and
$$\Delta=\Big(\cot\frac{\pi}{2\sqrt{2(n-1)}}+1\Big)n.$$
\end{thm}
\begin{proof}
By Construction 1, it is easy to see that $H$ satisfies condition
(1) of Theorem \ref{thm-Con-Frmk} and $t=\sqrt{\frac{n-1}{2}}$. By
Lemma \ref{lem-Cnstrn-bnd}, we have $|\langle\bm u, \bm
u'\rangle|\leq\rho=\cos\frac{\pi}{2t}=\cos\frac{\pi}{\sqrt{2(n-1)}}$
for all distinct $\bm u,\bm u'\in\Omega$. Therefore, by Theorem
\ref{thm-Con-Frmk}, we have $$\Gamma_{2}(\mathcal C)\leq
\frac{2n}{\sqrt{1-\cos^2\frac{\pi}{\sqrt{2(n-1)}}}}
=\frac{2n}{\sin\frac{\pi}{\sqrt{2(n-1)}}}.$$ Moreover, by Theorem
\ref{thm-crt-1-err}, the decoder $\mathcal D_1$ for $\mathcal C$
can correct a single error with respect to the threshold pair
$(\delta, \Delta)$ such that $\delta=1$ and
\begin{align*}
\Delta&=\frac{\sqrt{1+\cos\frac{\pi}{\sqrt{2(n-1)}}}
+\sqrt{1-\cos\frac{\pi}{\sqrt{2(n-1)}}}}{\sqrt{1-\cos\frac{\pi}{\sqrt{2(n-1)}}}}n\\
&=\Big(\cot\frac{\pi}{2\sqrt{2(n-1)}}+1\Big)n
\end{align*} which completes the proof.
\end{proof}

In general, given any $n>3$, let
$t=\Big\lceil\sqrt{(n-1)/2}\Big\rceil$, then $n\leq 2t^2+1$. Let
$\Omega$ be obtained from Construction 1. By Lemma
\ref{lem-Cnstrn-bnd}, for all distinct $\bm u,\bm u'\in\Omega$, we
have
$$|\langle\bm u, \bm
u'\rangle|\leq\rho=\cos\frac{\pi}{2t}=\cos\frac{\pi}{2\Big\lceil\sqrt{(n-1)/2}\Big\rceil}.$$
Let $H$ be a $3\times n$ matrix consists of any given $n$ vectors
in $\Omega$ and $\mathcal C$ be the code that has $H$ as a parity
check matrix. Similar to Theorem \ref{Cntrn-rdncy-3}, we can
obtain $$\Gamma_{2}(\mathcal C)\leq
\frac{2n}{\sin\frac{\pi}{2\Big\lceil\sqrt{(n-1)/2}\Big\rceil}}=
O(n\sqrt{n})$$ and
$$\Delta=\Big(\cot\frac{\pi}{4\Big\lceil\sqrt{(n-1)/2}\Big\rceil}+1\Big)n=
O(n\sqrt{n}).$$ In fact, as $n\rightarrow\infty$, the limit
superiors of $\Gamma_{2}(\mathcal C)/(n\sqrt{n})$ and
$\Delta/(n\sqrt{n})$ are both bounded above by $2\sqrt{2}/\pi$.

\end{document}